\documentclass[preprint,5p,times,numbered]{elsarticle}

\usepackage{url}
\usepackage{amsmath,amsfonts}
\usepackage{algorithmic}
\usepackage{algorithm}
\usepackage{array}
\usepackage[caption=false,font=normalsize,labelfont=sf,textfont=sf]{subfig}
\usepackage{textcomp}
\usepackage{stfloats}
\usepackage{url}
\usepackage{verbatim}
\usepackage{graphicx}
\usepackage[dvipsnames]{xcolor}
\usepackage{amsthm} 
\newtheorem{Def}{Definition}
\newtheorem{Thm}{Theorem}

\begin{document}

\begin{frontmatter}


\title{Local Differential Privacy for Tensors in Distributed Computing Systems}

\author[1]{Yachao Yuan}
\ead{chao910904@suda.edu.cn}

\author[2]{Xiao Tang}
\ead{serenade\_b612@163.com}
\ead{Yachao Yuan and Xiao Tang have equal contributions.}

\author[3]{Yu Huang}
\ead{fatmo@nuaa.edu.cn}

\author[1]{Yingwen Wu*}
\ead{ywwu@suda.edu.cn}
\ead{Co-corresponding author}

\author[1]{Jin Wang*}
\ead{ustc\_wangjin@hotmail.com}
\ead{Co-corresponding author}

\affiliation[1]{organization={School of Future Science and Engineering, Soochow University},
            city={Suzhou},
            postcode={215000}, 
            state={Jiangsu},
            country={China}}
\affiliation[2]{organization={School of Mathematics, Southeast University},
            city={Nanjing},
            postcode={210000}, 
            state={Jiangsu},
            country={China}}
\affiliation[3]{organization={School of Cyber Science and Engineering, Southeast University},
            city={Nanjing},
            postcode={210000}, 
            state={Jiangsu},
            country={China}}

\begin{abstract}
Tensor-valued data, increasingly common in distributed big data applications like autonomous driving and smart healthcare, poses unique challenges for privacy protection due to its multidimensional structure and the risk of losing critical structural information. Traditional local differential privacy methods, designed for scalars and matrices, are insufficient for tensors, as they fail to preserve essential relationships among tensor elements. 
We propose TLDP, a novel \emph{LDP} algorithm for \emph{T}ensors that randomly decides whether to inject noise into each component with a certain probability. Such randomness substitutes for part of the perturbation noise and safeguards the raw data more effectively under the condition of guaranteeing differential privacy.
To strike a better balance between utility and privacy, we incorporate a weight matrix that selectively protects sensitive regions. Both theoretical analysis and empirical findings from real-world datasets show that TLDP achieves superior utility while preserving privacy, making it a robust solution for high-dimensional tensor data. 
\end{abstract}





\end{frontmatter}

\section{Introduction}
High-dimensional tensor data are passionately required in a diverse array of distributed big data applications, such as autonomous driving \cite{10907793}, smart healthcare \cite{WANG2025103084}, and face attribute recognition technology \cite{zhang2025privacy}. In these systems, tensor data that contains sensitive private information, such as location data, facial images, or medical records, is frequently transmitted between different parties, which creates a significant risk of privacy leakage. For example, through the intercepted tensor-valued features or model parameters (in federated systems), an attacker can easily infer whether a data type was used for training the model. Typically, Local Differential Privacy (LDP) constrains an attacker's ability to infer details about individual data from the shared data.

The conventional approach to dealing with tensors involves treating them either as a set of scalars or as a group of vectors to ensure local differential privacy.
However, unlike scalar or vector data, tensor data frequently maintains the inherent structure of data, such as multimedia content and graph-based information. For instance, video data contains correlated sequential frames, while graph data consists of critical structural and relational information. Critical structural information of the tensor-valued data and interrelationships among elements could be lost by flattening them into plain vectors or treating them as scalars when protecting privacy. For example, a time-series tensor like videos in computer vision tasks or financial data in stock market analysis applications contain crucial instructional information, and flattening them goes against common sense; or consider a symmetric tensor like Diffusion Tensor Imaging (DTI) in medical imaging or the polarizability tensor in chemistry, if each element is perturbed using traditional methods, such as adding independent and identically distributed noise, the probability that the tensor still remains symmetric is extremely low. 
In this study, we establish a formal definition of LDP for tensors and develop approaches to satisfy this definition.

Developing schemes for tensors to achieve local differential privacy is a very challenging task due to the complexity of high-dimensional data. 
Generally, the noise introduced to the raw data is deeply affected by the sensitivity of the query function, one key factor of which is the size of the tensor. Essentially, the larger the specifications of the data tensors, the more substantial the noise tends to be, potentially leading to a decrease in the accuracy of the results.
Extensive research explored LDP schemes for scalars~\cite{warner1965randomized,rappor,fanti2015building}, vectors~\cite{10726610,qiu2023fast}, and matrices~\cite{zheng2022matrix,individual} by either adding random noise or randomly perturbing individual inputs to achieve local differential privacy and achieved promising results, but an LDP approach applicable to tensor-valued data in distributed computing systems has yet to be established. 
Although a few research studies like~\cite{LoPub} introduce an LDP for tensors for distribution estimation, they cannot be directly applied for locally differential privacy protection in distributed computing systems. It is because although the perturbed tensor-valued data still retain certain statistical properties, the model analysis results on such data will be seriously affected. Therefore, in this work, we aim to tackle this challenge by developing a local differential privacy mechanism that ensures the effectiveness of the subsequent data analytics despite the high dimensionality of tensors.

Both privacy and utility play a critical role when practitioners apply LDP in practical distributed computing systems. However, designing an LDP that balances the privacy and utility of data presents a significant challenge due to the inherent trade-off between the two. This issue is even more pronounced with tensors, where excessive noise or other perturbations are usually added due to the high dimensionality of tensors, which can severely diminish their usefulness. 
To mitigate such degradation, various methods have been explored, as seen in~\cite{sun2021privbv,bao2021privacy,zhang2024ai,liu2025federated}. 
However, most of them are for scalars or vectors instead of tensors, lacking a utility guarantee for tensors in distributed computing systems.
In contrast, our approach ensures practical utility while satisfying $\epsilon-$LDP in distributed computing systems.

In this paper, we formalize the definition of LDP for tensor-valued data and innovatively introduce a \emph{L}ocal \emph{D}ifferential \emph{P}rivacy algorithm designed for \emph{T}ensors of varying dimensions named TLDP, being able to preserve the tensor-valued data's structural information. 
TLDP incorporates locally differentially private and tensor-shaped noise into the tensors by randomly selecting components of the tensor with a randomized response mechanism to accomplish local differential privacy. The key idea behind it is to randomly choose whether or not to add Laplace/Gaussian noise to the tensor data so that the total amount of noise added to the tensors is lower than directly adding noise to each entity of the tensor without breaking the condition of satisfying $\epsilon-$LDP. We rigorously prove that TLDP satisfies local differential privacy both when using the Laplace/Gaussian noise and has a lower expected error than previous studies.

Additionally, it is easy to observe that different parts of a tensor may have varying sensitivities in practical applications. For example, in a road scene image, the facial or license plate number information is more critical than the background for scene understanding tasks in autonomous driving applications. The data privacy of such tasks can be further improved by incorporating a precisely tailored weight matrix that indicates the areas of the tensor that are more important than others. Therefore, we design an optimized local differential privacy mechanism based on TLDP by applying varying weights to different parts of the tensor.

The primary contributions of this work include:
\begin{itemize}
\item We propose an algorithm that ensures $\epsilon-$LDP for tensor-valued data named TLDP. TLDP has a lower expected error and greater utility compared to prior methods, a property independent of the tensor's shape. 
To the best of our knowledge, it is the first work that addresses local differential privacy of tensor-valued data in machine learning-based distributed computing systems.
\item We further enhance the algorithm by incorporating a weight matrix, which selectively improves the impact of perturbation on important portions of the data, thereby strengthening privacy protection. 
\item We provide rigorous theoretical proof of privacy under the TLDP mechanisms and present a series of experiments performed on various tensor-valued models and datasets. The results demonstrate that our method achieves superior utility compared to other mechanisms while maintaining the same level of privacy.
\end{itemize}

\section{Related Work} \label{sec:RelatedWork}
Our research aligns closely with studies in the following areas.
\subsection{Differential Privacy for Scalar Data}
\subsubsection{Differential privacy by adding random noise}

In differential privacy, the Gaussian mechanism \cite{dwork2014algorithmic} and the Laplace mechanism \cite{dwork2006calibrating} are frequently employed as privacy-preserving mechanisms. Although they do not directly apply to tensor-valued queries, our research is nevertheless grounded in the fundamental principle of preserving privacy through the injection of noise.
The Gaussian mechanism introduced independent and identically distributed (i.i.d.) Gaussian noise, calibrated according to the $l_2$-sensitivity of the query, guarantees $(\epsilon, \delta)$-differential privacy. Similarly, by incorporating noise derived from the Laplace distribution, scaled according to the query function's $l_1$-sensitivity, the Laplace mechanism enforced strong $\epsilon$-differential privacy. 

Beyond these fundamental mechanisms, several sophisticated approaches have also been introduced.
Mechanisms that derive their privacy guarantees from the above mechanisms are well-established, as exemplified by the composition schemes referenced in~\cite{kairouz2014extremal,lee2018concentrated,yu2019differentially,lecuyer2019privacy}. Among these, Kairouz et al.~\cite{kairouz2014extremal} proposed an optimal composition scheme that is applicable to a broader range of noise distributions. The studies in~\cite{lee2018concentrated,yu2019differentially} provided dynamic accounting methods that adjust based on the algorithm's runtime convergence. Moreover, Lecuyer et al.~\cite{lecuyer2019privacy} ensured a global differential privacy guarantee in environments where datasets are continuously expanding.

\subsubsection{Local differential privacy by randomized response to a query}
Warner et al.~\cite{warner1965randomized} introduced the randomized response mechanism in 1965, which is a technique specifically designed to preserve respondent privacy, commonly referred to as $w$-RR. The key idea is to introduce randomness into the response process, allowing respondents to answer sensitive questions without fear of exposure, while still enabling researchers to estimate the true distribution of responses.
The $w$-RR technique is specifically designed for discrete datasets with exactly two distinct values. 
To extend its applicability, two primary paths for improvement can be explored. The first involves encoding and transforming the variable's multiple values, as demonstrated in methods such as Randomized Aggregatable Privacy-Preserving Ordinal Response (RAPPOR)~\cite{rappor} and S-Hist~\cite{bassily2015local}, to ensure compatibility with the binary nature expected by $w$-RR.
The RAPPOR employs Bloom filters and hash functions to encode a single element into a vector, then perturbs the original value through both permanent and instantaneous randomized responses. It represents a method for collecting statistical data from end-user applications while preserving anonymity and ensuring robust privacy protections. 
However, RAPPOR has two main drawbacks: (1) the transmission overhead between the user and the data aggregator is relatively high, as each user is required to send a vector whose length is determined by the size of its bloom filter; (2) the data collector must pre-collect a list of candidate strings for frequency counting. 

To address the first issue, the S-Hist method~\cite{bassily2015local} offered a solution in which each user encodes the string, randomly selects a bit, applies random response techniques to perturb it, and then sends it to the data collector. This approach significantly reduces the transmission cost by minimizing the amount of data each user needs to send. 
For the second problem, building upon the RAPPOR-based encoding-decoding framework, Kairouz et al.~\cite{kairouz2016discrete} further proposed the $O$-RR method by introducing hash mapping and grouping operations. Hash mapping allows the method to focus on the encoded values rather than the original strings, eliminating the need to pre-collect a list of candidate strings. Additionally, the use of grouping operations reduces the probability of hash mapping value collisions. This combination of techniques enhances the robustness and flexibility of differential privacy in various data collection scenarios.
To overcome the second limitation of $w$-RR, one approach is to improve the distribution of the $w$-RR technique, making it applicable to variables with more than two values. This has been achieved through mechanisms such as $k$-RR~\cite{kairouz2014extremal}. In the $k$-RR method, the original data is preserved with probability $p$ and is flipped to any other value with a probability of $\frac{1-p}{n-1}$,  where $n$ is the number of possible values. This mechanism ensures that the perturbation is uniformly distributed across all possible values, thereby preserving the privacy of the data. 


\subsection{Differential Privacy for Non-scalar Data}
Recently, as the application of differential privacy has expanded to more scenarios, differential privacy methods designed for non-scalar datasets, i.e., vectors, matrices, and tensors, have been proposed.
For instance, Li et al. \cite{10726610} proposed a new privacy budget recovery mechanism, called ChainDP, which adds noise sequentially, so that the noise added by the previous user can be used by subsequent users, thereby reducing the overall noise level and estimation error. The MVG mechanism \cite{chanyaswad2018mvg} is tailored for matrix-valued queries and incorporates matrix-valued noise to ensure $(\epsilon, \delta)$-differential privacy. The $l_2$-sensitivity of the MVG mechanism is measured using the Frobenius norm, which quantifies the discrepancy between two neighboring matrices. 
The IDN mechanism \cite{direction} achieved differential privacy by appending a noise tensor of equal magnitude to the data tensor. Abadi et al.~\cite{abadi2016deep} presented an innovative accounting approach for the Gaussian mechanism that minimizes the total additive noise while preserving the same level of privacy protection.
Moreover, in order to obtain more utility under the same privacy budget, Zhang et al. \cite{LLAGP} proposed an adaptive gradient perturbation mechanism, injecting different amount of noise according to the order of hidden layers, and Tan et al. \cite{DPHSGD} used a diagonal Hessian Matrix to generate a more effective clipping threshold of gradients.

However, LDP demands a higher level of privacy than traditional differential privacy, necessitating more noise to achieve comparable privacy guarantees, which can significantly degrade data utility due to extensive perturbation; thus additional mechanisms need to be introduced to solve the issue. 
Qiu et al. \cite{qiu2023fast} introduced a sparse weight matrix that computes the non-zero elements of the matrix, thereby reducing the computational load of LDP when applied to large-scale datasets. But neither of them is specific to tensors.
Wang et al. \citep{individual} applied differential privacy to high-order, high-dimensional sparse tensors within the IoT transmission context through individual randomized responses. But when the central server is required to receive tensor information directly from users, this approach risks disclosing the user's original information with a certain probability, thereby undermining the purpose of privacy protection. Ren et al. \cite{LoPub} proposed a locally privacy-preserving scheme for crowdsensing systems to collect and build high-dimensional data from distributed users. It is improved based on the RAPPOR algorithm, which uses a hash function to convert raw data into fixed-length binary feature vectors. Due to the irreversibility of hash functions (i.e., the inability to reversely derive original attribute values from the bitstring of a Bloom filter), this method is not applicable to the scenarios in this paper where raw dataset elements are required and the query function is the identity function.


\subsection{Learning With Differential Privacy}
A range of studies explored the integration of differential privacy techniques into machine learning models, addressing various privacy-preserving objectives. These efforts encompass differentially private individual data~\cite{chanyaswad2018mvg,wang2018not}, model outputs~\cite{papernot2018scalable}, model parameters~\cite{lee2018concentrated,lecuyer2019privacy,yu2019differentially,song2013stochastic,shokri2015privacy,bassily2014private,wang2017differentially}, and even objective functions~\cite{phan2016differential,zhang2017efficient}. Our focus is especially on machine learning applications in distributed systems, which often involve frequent transmission of high-dimensional or complex tensor-valued data like raw private data, extracted features, or model parameters. 

In summary, we consider our paper introduces a primitive local differential privacy scheme, which fits well in distributed computing systems but certainly can be applied in any other machine learning systems, where individual users' tensor-valued data is private and requires preservation.

\section{Preliminaries} \label{sec:Preliminaries}
In this section, we aim to provide our readers with foundational knowledge of differential privacy to facilitate the comprehension of our work.

\subsection{Local Differential Privacy}
To address the scenario where the central server is untrustworthy, local differential privacy (LDP) enables data collection and analysis while ensuring strong privacy protection by perturbing data before it is shared with a central aggregator. This ensures that even the data collector cannot determine an individual’s true value with certainty.
The privacy guarantee of LDP is expressed by bounding the logarithmic ratio of the output probabilities for any two possible inputs from a single user, ensuring that the private mechanism $M$ does not reveal too much about the user's true input. This local and individualized privacy guarantee makes LDP particularly suitable for scenarios where users do not trust the data collector and need to protect their data before sharing it. Formally, given any two tensor-valued inputs from the same user $\mathcal{X}_1$ and $\mathcal{X}_2$, $y$ be the output and $M$ be the private scheme, the following condition holds:
\begin{Def}[$\epsilon$-LDP]\label{def1}
    An randomized algorithm \textbf{M} satisfies $\epsilon$-LDP (where \(\epsilon \geq 0\)), if for any tensor-valued inputs $\mathcal{X}_1$ and $\mathcal{X}_2$ from the same user, we have $$\forall y\in Range(\textbf{M}), Pr[\textbf{M}(\mathcal{X}_1)=y] \leq e^{\epsilon}\cdot  Pr[\textbf{M}(\mathcal{X}_2)=y]$$ where $Range(\textbf{M})$ represents a set of all possible outputs of the algorithm \textbf{M}.
\end{Def}

The parameter $\epsilon$ is the privacy budget of the protection mechanism \textbf{M}, representing the protection level of data. Generally speaking, a lower budget $\epsilon$ implies a higher level of privacy protection but usually lower data utility. So the core of research on DP and LDP lies in how to enhance data utility on the premise of ensuring privacy protection, thereby striking a balance between privacy and utility.

\textbf{Problem Definition and Notations.} We suppose that in this article, each user has a number of tensor-valued data, and $\mathcal{X}\in\mathbf{R}^{I_1\times I_2\times\dots\times I_N}$ is any one of them. We utilize $\Delta$ to represent the range of the values of a tensor in a dataset and $[\Delta]$ to denote all possible values of the tensor. For example, $\Delta=256$ and $[\Delta]=\{0, 1, 2, \dots, 255\}$ for an image of the MNIST dataset, because that range of values of an image of the MNIST dataset is 0-255.Without compromising generality, we assume the input domain is $[\Delta]$, i.e., $\mathcal{X}_{i_1 i_2 \dots i_N}\in[\Delta]$. If the tensor takes continuous values, it is recommended to discretize it initially.
The notations used in this paper are summarized in Table~\ref{tab:table1} for easy reference.
\begin{table}[t]
\caption{Notations}\label{tab:table1}
\centering
\resizebox{\columnwidth}{!}{
\begin{tabular}{c c}
\hline
Symbol & Definition\\
\hline
$\mathcal{X}$ & N-order tensor \\
$\mathcal{Z}$ & Noise tensor \\
$I$ & Product of the tensor's dimensions, i.e., $I=I_1 I_2 \cdots I_N$ \\
$\Omega$ & Index $(i_1,i_2,\dots,i_N)$ of tensor components\\
$\Delta$ & Range of the values of a tensor in a dataset\\
$\epsilon$ & Privacy budget\\
$Lap(\mu,b)$ & Laplace distribution with mean $\mu$ and variance $2b^2$\\
$\mathcal{N}(\mu,\sigma^2)$ & Normal distribution with mean $\mu$ and variance $\sigma^2$\\
$p$ & Flipping probability\\
\hline
\end{tabular}
}
\end{table}

In the context of LDP application, privacy protection is moved downstream from central servers to data-generating endpoints, such as users or edge devices. This shift eliminates the need for a trusted third party to handle raw data, as sensitive information is processed locally before any transmission or sharing. Therefore, our goal is to perturb the original private tensor on the client-side, ensuring that users' sensitive information can be properly preserved without compromising data utility. To access the level of privacy protection, the parameter $\epsilon$, defined as the upper bound of the KL-divergence between the data distributions before and after perturbation, quantifies the effectiveness of privacy-preserving algorithms. As stated in Definition~\ref{def1}, a lower value of $\epsilon$ corresponds to stronger privacy guarantees. To access the utility, we evaluate the F1-score of the model parameters trained from the perturbed data as a measure of the data's effectiveness. Specifically, we train the model with the perturbed data and then validate it with the unperturbed data. The closer the validation results are to the original, the higher the utility. It is evident that as privacy protection increases (i.e., $\epsilon$ decreases), the accuracy of the trained model tends to decline. Our research seeks to strike a balance between data utility and privacy protection.

\subsection{Randomized Response}
Randomized response is a statistical technique introduced in the 1960s by \cite{warner1965randomized} for gathering data on sensitive topics while allowing respondents to maintain confidentiality. Surveys employing this method enable the accurate calculation of precise population statistics while ensuring individual privacy. 
Today, it is the dominant perturbation mechanism used in local differential privacy protection of scalars. 
For binary-valued problems, the randomized response method retains the original data with probability \( p \) and flips it with probability \( 1 - p \). In most cases, $p>1/2$, i.e., $p>1-p$, so for any $v_1 \neq v_2$, we have $$\frac{Pr[v_1=y]}{Pr[v_2=y]}\leq \frac{p}{1-p}.$$ Therefore, it satisfies \(\ln(\frac{p}{1-p})\)-LDP.

\section{Local Differential Privacy of Tensors} \label{sec:methodologies}
In this section, we will formally introduce the proposed local differential privacy method for tensors. First, we present the unified definition of LDP for tensor-valued data, then the motivation, theorem, and proof of TLDP and weighted TLDP are presented.\\

\begin{Def}[\textit{LDP for Tensor-valued Data}]\label{def2}
    A randomized mechanism \textbf{M} gives $\epsilon$-LDP for tensors, where $\epsilon \geq 0$, if and only if for any data records $\mathcal{X}$ and $\mathcal{X}'$, we have $$\forall \mathcal{X}^\ast \in Range(\textbf{M}), Pr[\textbf{M}(\mathcal{X})=\mathcal{X}^\ast] \leq e^{\epsilon}\cdot Pr[\textbf{M}(\mathcal{X}')=\mathcal{X}^\ast],$$ where $Range(\textbf{M})$ represents the set of all possible outputs of the algorithm \textbf{M}.
\end{Def}

When DP is applied for privacy protection, the indicator global sensitivity, defined as $\Delta f=\underset{D\sim D'}{max}\,||f(D)-f(D')||$, which quantifies the maximum change in query results between two adjacent datasets, is indispensable to determine the noise level of privacy protection. However, in Local Differential Privacy (LDP), the concept of global sensitivity does not apply since users are unaware of each other's records. In this paper, we use $\Delta$ to denote the range of variation of the raw data. Although $\Delta$ and the global sensitivity are numerically equal as the query function is specified to be identity in this paper, it is worth noting that $\Delta$ differs from global sensitivity: the former is only related to the users' initial data, while global sensitivity pertains to query functions and is affected by the size of tensors.

Comparing Definition~\ref{def1} with Definition~\ref{def2}, it is evident that Definition~\ref{def1} is a special case of Definition~\ref{def2}. Specifically, when the tensor order is reduced to one, the multi-dimensional tensor structure in Definition~\ref{def2} degenerates into a one-dimensional vector form, which aligns exactly with the data representation specified in Definition~\ref{def1}. In contrast, Definition~\ref{def2} relaxes the constraint on tensor order, allowing $N$ to take any positive integer value and supporting multi-indexed tensor inputs (e.g., $N = 2$ for matrix-valued data and $N = 3$ for 3D tensor data) by introducing the index set $\{I_1, I_2, ..., I_N\}$. This hierarchical relationship between the two definitions is not only theoretically consistent but also practically valuable: it means that any conclusion derived under Definition~\ref{def1} can be naturally extended to the broader framework of Definition~\ref{def2}, while Definition~\ref{def2} retains the ability to handle simpler one-dimensional scenarios through its special case.

\subsection{TLDP}
We introduce TLDP to protect the privacy of a user's tensor data in distributed systems before transmitting it to untrusted parties. This approach enables the aggregation of collective data while preserving individual privacy, eliminating the need for a trusted central server.

Originally designed for scalar data with binary outcomes, the randomized response technique has been extended to multi-valued scenarios by encoding scalars into higher-dimensional vectors or adjusting flipping probabilities.
However, applying these methods to tensor data presents unique challenges. Encoding each tensor element individually would require substantial storage resources, exponentially increasing algorithmic complexity. Conversely, given our research's need to retain information from each tensor rather than relying solely on aggregation, methods that randomly perturb selected tensor elements could significantly reduce the utility of the original data. Both refinement strategies are unable to meet our requirements of perturbing the raw tensors as little as possible while preserving differential privacy and avoiding other limiting factors such as excessive storage space.

In light of these challenges, there is a clear need for more effective methods that balance privacy protection and the utility of tensor data. Building upon the aforementioned concepts, we propose a novel approach named TLDP for achieving tensor differential privacy by injecting noise into specific positions of tensors, which are determined via randomized response. In this method, each tensor element retains its original value with a probability of $p$ and introduces calibrated noise (either Laplace or Gaussian noise) with a probability of $1-p$. The randomized response mechanism still determines whether a piece of data will be altered, except that the way of alteration is no longer flipping but injecting noise. The overall procedure is shown in Algorithm \ref{alg:alg1}. Our method ensures privacy protection for each user and demonstrates that, compared to directly adding a noise tensor, the inherent randomness of randomized response allows for meeting privacy budgets without requiring excessive noise, thereby maintaining a balance between privacy preservation and data utility. In the following, we will employ rigorous mathematical proofs to verify the privacy preservation property of this algorithm.

\begin{algorithm}[t]
    \caption{TLDP}\label{alg:alg1}
    $\textbf{Input:}$ (a) privacy parameter $\epsilon$,\quad(b) range of the initial data $\Delta$,\quad (c) raw tensor $\mathcal{X}\in\mathbb{R}^{I_1\times I_2\times\dots\times I_N}$\\
    $\textbf{Output:}$ perturbed tensor $\widetilde{\mathcal{X}}$
    \begin{algorithmic}[1]
        \STATE Initialize noise tensor $\mathcal{Z}\in\mathbb{R}^{I_1\times I_2\times\dots\times I_N}$
        \STATE Compute Laplace noise parameter $b=\frac{\Delta}{\epsilon}$, retaining possibility $p=\frac{e^{\epsilon-\frac{I\Delta}{b}}}{2b+e^{\epsilon-\frac{I\Delta}{b}}}$; or Gaussian noise parameter $\sigma^2=\frac{\Delta^2}{2\epsilon}$, retaining possibility $p=\frac{e^{\epsilon-\frac{I\Delta^2}{2\sigma^2}}}{\sigma\sqrt{2\pi}+e^{\epsilon-\frac{I\Delta^2}{2\sigma^2}}}$
        \FOR{$(i_1,i_2,\dots,i_N)\in\Omega$}
            \STATE Generate a random number $r\in (0,1)$
            \IF{$r> p$}
                \STATE Generate noise that correspond to Laplace distribution $\mathcal{Z}_{i_1,i_2,\dots,i_N}\sim Lap(0,b)$ or Gaussian distribution $\mathcal{Z}_{i_1,i_2,\dots,i_N}\sim \mathcal{N}(0,\sigma^2)$
            \ELSIF{}
                \STATE $\mathcal{Z}_{i_1,i_2,\dots,i_N}=0$
            \ENDIF
        \ENDFOR
        \RETURN $\widetilde{\mathcal{X}}=\mathcal{X}+\mathcal{Z}$
    \end{algorithmic}
\end{algorithm}

\begin{Thm}[Privacy of TLDP]
    TLDP ensures the LDP of tensor information.
\end{Thm}
    
\begin{proof}
    Let $\mathcal{X}$ and $\mathcal{X}'$ be any two tensors of the same size from the same user, and $\mathcal{X}^\ast$ be any output. It can be observed that the probability of the tensor remaining unperturbed is the highest. 
    We assume that in the unperturbed tensor $\mathcal{X}$, there are a total of t components that preserve the original data, because the corresponding random numbers are smaller than p.\\
        (1)\; Laplace noise\\
        Since a Laplace distribution with mean 0 and scale parameter $b$,
        \begin{equation}
            Lap(x|b)=\frac{1}{2b}e^{-\frac{|x|}{b}},
        \end{equation}
        so 
    \begin{equation}
        \begin{aligned}
       &\frac{Pr[\mathcal{M}_{Lap}(\mathcal{X})=\mathcal{X}^\ast]}{Pr[\mathcal{M}_{Lap}(\mathcal{X}')=\mathcal{X}^\ast]}\\
       & \leq \frac{p^t[(1-p)\cdot\frac{1}{2b}]^{I-t}}{[(1-p)\cdot \frac{1}{2b}\cdot e^{-\frac{\Delta}{b}}]^I}
       =(\frac{p}{1-p}\cdot 2b)^t\cdot e^{\frac{I\Delta}{b}}\\
       & \leq\frac{p}{1-p}\cdot 2b\cdot e^{\frac{I\Delta}{b}} =\frac{e^{\epsilon-\frac{I\Delta}{b}}}{2b}\cdot 2b\cdot e^{\frac{I\Delta}{b}}\\
       &=e^{\epsilon}.  
    \end{aligned}
    \end{equation}
        (2)\;Gaussian noise\\
        Since a Gaussian distribution with mean 0 and scale parameter $\sigma^2$,
        \begin{equation}
            N(x|\sigma^2)=\frac{1}{\sigma\sqrt{2\pi}}e^{-\frac{x^2}{2\sigma^2}},
        \end{equation}
        so 
    \begin{equation} \label{eq:TLDP-G}
        \begin{aligned}
       &\frac{Pr[\mathcal{M}_{Gau}(\mathcal{X})=\mathcal{X}^\ast]}{Pr[\mathcal{M}_{Gau}(\mathcal{X}')=\mathcal{X}^\ast]}\\
       &  \leq \frac{p^t[(1-p)\cdot\frac{1}{\sigma\sqrt{2\pi}}]^{I-t}}{[(1-p)\cdot \frac{1}{\sigma\sqrt{2\pi}}\cdot e^{-\frac{\Delta^2}{2\sigma^2}]^I}}
       =(\frac{p}{1-p}\cdot \sigma\sqrt{2\pi})^t\cdot e^{\frac{I\Delta^2}{2\sigma^2}}\\
       & \leq\frac{p}{1-p}\cdot \sigma\sqrt{2\pi} \cdot e^{\frac{I\Delta^2}{2\sigma^2}}
       =\frac{e^{\epsilon-\frac{I\Delta^2}{2\sigma^2}}}{\sigma\sqrt{2\pi}}\cdot \sigma\sqrt{2\pi}\cdot e^{\frac{I\Delta^2}{2\sigma^2}}=e^{\epsilon}.  
    \end{aligned}
    \end{equation}
\end{proof}

The penultimate equations in the above proofs are derived by substituting the formula for $p$ from Algorithm \ref{alg:alg1}. The same logic applies to subsequent proofs, and no further explanation will be provided hereafter.

Although DP and LDP differ in several respects, when both perturbation methods involve injecting Laplace noise or Gaussian noise, the amount of noise under the same $\epsilon$ can reflect the impact of each on the accuracy of the results. The following Theorem \ref{thm:important} demonstrates that, under the same privacy budget, the TLDP method can significantly reduce the amount of noise injected compared to adding noise to all components.

\begin{Thm}\label{thm:important}
    Compared to methods that inject noise into each component, LDP for tensors reduces the noise while guaranteeing the same privacy budget.
\end{Thm}

\begin{proof}
    Let $\mathcal{X}$ and $\mathcal{X}'$ be any two tensors of the same size, and $\mathcal{X}^\ast$ be any output. 
    \begin{enumerate}
    \item{Laplace noise}\\
    If the noise with mean 0 and scale parameter $b_1$ is injected into each component:
    \begin{equation}
    \begin{aligned}
        \frac{Pr[\mathcal{M}_{Lap}(\mathcal{X})=\mathcal{X}^\ast]}{Pr[\mathcal{M}_{Lap}(\mathcal{X}')=\mathcal{X}^\ast]}&\leq\frac{(\frac{1}{2b_1})^I}{(\frac{1}{2b_1}\cdot e^{-\frac{\Delta}{b_1}})^I}=e^{\frac{I\Delta}{b_1}}.
    \end{aligned}
    \end{equation}
    Let $e^{\frac{I\Delta}{b_1}}=e^{\epsilon}$, then $b_1=\frac{I\Delta}{\epsilon}$. Given the same privacy budget and sensitivity, we have 
    \begin{equation}
        \frac{b}{b_1}=\frac{1}{I}.  
    \end{equation}
    Combining with the feature of Laplace distribution that the larger the scale parameter \( b \), the greater the noise, we conclude that TLDP injects less Laplace noise while guaranteeing the same privacy budget as adding noise tensor directly. In particular, the larger the size of the data, the more significant the noise reduction effect will be.

    \item{Gaussian noise}\\
    If the noise with mean 0 and scale parameter $\sigma_1^2$ is injected into each component:
    
    \begin{equation}
    \begin{aligned}
        \frac{Pr[\mathcal{M}_{Gau}(\mathcal{X})=\mathcal{X}^\ast]}{Pr[\mathcal{M}_{Gau}(\mathcal{X}')=\mathcal{X}^\ast]}&\leq\frac{(\frac{1}{\sigma_1\sqrt{2\pi}})^{I}}{(\frac{1}{\sigma_1\sqrt{2\pi}}\cdot e^{-\frac{\Delta^2}{2\sigma_1^2}})^{I}}=e^{\frac{I\Delta^2}{2\sigma_1^2}}.
    \end{aligned}
    \end{equation}
    Let $e^{\frac{I\Delta^2}{2\sigma_1^2}}=e^{\epsilon}$. Given the same privacy budget and sensitivity, we have 
    \begin{equation}
        \frac{\sigma^2}{\sigma^2_1}=\frac{1}{I}.
    \end{equation}
    Considering the property of the Gaussian distribution that a larger scale parameter $\sigma^2$ results in greater noise, we conclude that TLDP injects less Gaussian noise while guaranteeing the same privacy budget compared to directly adding a noise tensor. Specifically, as the size of the data increases, the noise reduction effect becomes more pronounced.
    \end{enumerate}
\end{proof}

\subsection{Weighted TLDP}
The overzealous pursuit of privacy protection often results in significant compromises to data usability. Striking a balance between these two facets is crucial, especially considering that data usability is not a variable we can easily manipulate. Therefore, attention must be focused on fine-tuning the privacy protection algorithms. Our innovative proposal involves assigning specific weights to the tensor components, where a higher weight indicates greater sensitivity at that particular position, consequently reducing the likelihood of its preservation. To simplify the algorithm, we introduce a weight matrix into the tensor data to be processed rather than using a weight tensor. For example, a color image is a three-order tensor of size $A\times B\times3$, and a weight matrix of size $A\times B$ is assigned based on the image content. Additional details are presented in Algorithm \ref{alg:alg2}. The weight matrix is given prior to the implementation of the algorithm, and there are different generation methods for it across various practical application scenarios. To illustrate, face recognition detection algorithms are used to identify facial regions in datasets such as LFW \cite{LFW}
, and these regions are assigned higher weights. Upon analyzing the following proof, it is evident that, given the same $\epsilon$, the presence of the weight matrix reduces the probability of retaining the original data compared to Algorithm~\ref{alg:alg1}, thereby providing stronger privacy protection.

\begin{algorithm}
    \caption{Weighted TLDP}\label{alg:alg2}
    \textbf{Input:} (a) privacy parameter $\epsilon$,\quad(b) range of the initial data $\Delta$,\quad (c) raw tensor $\mathcal{X}\in\mathbb{R}^{I_1\times I_2\times\dots\times I_N}$, (d)weight matrix $W=(w_{ij})_{I_M\times I_N}$\\
    \textbf{Output:} perturbed tensor $\widetilde{\mathcal{X}}$
    \begin{algorithmic}[1]
        \STATE Initialize noise tensor $\mathcal{Z}\in\mathbb{R}^{I_1\times I_2\times\dots\times I_N}$
        \STATE Compute Laplace noise parameter $b=\frac{\Delta}{\epsilon}$ or Gaussian noise parameter $\sigma^2=\frac{\Delta^2}{2\epsilon}$
        \FOR{$(i_1,i_2,\dots,i_N)\in\Omega$} 
            \STATE Compute the retaining possibility $p=\frac{(1-w_{mn})e^{\epsilon-\frac{I\Delta}{b}}}{2b+e^{\epsilon-\frac{I\Delta}{b}}}$ (Laplace) or $p=\frac{(1-w_{mn})e^{\epsilon-\frac{I\Delta^2}{2\sigma^2}}}{\sigma\sqrt{2\pi}+e^{\epsilon-\frac{I\Delta^2}{2\sigma^2}}}$ (Gaussian), where $w_{mn}$ is the weight corresponding to $\mathcal{X}_{i_1i_2\dots i_N}$
            \STATE Generate a random number $r\in (0,1)$
            \IF{$r> p$}
                \STATE Generate noise that correspond to Laplace distribution $\mathcal{Z}_{i_1,i_2,\dots,i_N}\sim Lap(0,b)$ or Gaussian distribution $\mathcal{Z}_{i_1,i_2,\dots,i_N}\sim N(0,\sigma^2)$
            \ELSIF{}
                \STATE $\mathcal{Z}_{i_1,i_2,\dots,i_N}=0$
            \ENDIF
        \ENDFOR
        \RETURN $\widetilde{\mathcal{X}}=\mathcal{X}+\mathcal{Z}$
    \end{algorithmic}
\end{algorithm}

\begin{Thm}[Privacy of Weighted TLDP]
    Weighted TLDP ensures the LDP of tensor information.
\end{Thm}

\begin{proof}
    Let $\mathcal{X}$ and $\mathcal{X}'$ be any two tensors of the same size, and $\mathcal{X}^\ast$ be any output. It can be observed that the probability of the tensor remaining unperturbed is the highest. 
    We assume that in the unperturbed tensor $\mathcal{X}$, there are a total of t components that preserve the original data, because the corresponding random numbers are smaller than p.
    \begin{enumerate}
        \item {Laplace noise}\\
        Since a Laplace distribution with mean 0 and scale parameter $b$,
        \begin{equation}
            Lap(x|b)=\frac{1}{2b}\cdot e^{-\frac{|x|}{b}},
        \end{equation}
        so
        \begin{equation}\label{10}
        \begin{aligned}
        &\frac{Pr[\mathcal{M}_{Lap}(\mathcal{X})=\mathcal{X}^\ast]}{Pr[\mathcal{M}_{Lap}(\mathcal{X}')=\mathcal{X}^\ast]}\\
       & \leq \frac{p^t[(1-p)\cdot\frac{1}{2b}]^{I-t}}{[(1-p)\cdot \frac{1}{2b}\cdot e^{-\frac{\Delta}{b}}]^I}
       =(\frac{p}{1-p}\cdot 2b)^t\cdot e^{\frac{I\Delta}{b}}\\
       & \leq\frac{p}{1-p}\cdot 2b\cdot e^{\frac{I\Delta}{b}}\\ 
       &= \frac{(1-w_{mn})\cdot e^{\epsilon-\frac{I\Delta}{b}}}{2b+w_{mn}\cdot e^{\epsilon-\frac{I\Delta}{b}}}\cdot 2b\cdot  e^{\frac{I\Delta}{b}}\\
       &\leq \frac{(1-w_{mn})\cdot e^{\epsilon-\frac{I\Delta}{b}}}{2b}\cdot 2b\cdot  e^{\frac{I\Delta}{b}}\\
       &=(1-w_{mn})\cdot e^{\epsilon}<e^{\epsilon}.
    \end{aligned}
    \end{equation}
    \item{Gaussian noise}\\
    Since a Gaussian distribution with mean 0 and scale parameter $\sigma^2$,
    \begin{equation}
        N(x|\sigma^2)=\frac{1}{\sigma\sqrt{2\pi}}\cdot e^{-\frac{x^2}{2\sigma^2}},
    \end{equation}
    so 
    \begin{equation}\label{12}
        \begin{aligned}
       &\frac{Pr[\mathcal{M}_{Gau}(\mathcal{X})=\mathcal{X}^\ast]}{Pr[\mathcal{M}_{Gau}(\mathcal{X}')=\mathcal{X}^\ast]}\\
       & \leq \frac{p^t[(1-p)\cdot\frac{1}{\sigma\sqrt{2\pi}}]^{I-t}}{[(1-p)\cdot \frac{1}{\sigma\sqrt{2\pi}}\cdot e^{-\frac{\Delta^2}{2\sigma^2}}]^I}\\
       &=(\frac{p}{1-p}\cdot \sigma\sqrt{2\pi})^t\cdot e^{\frac{I\Delta^2}{2\sigma^2}}\\
       & \leq\frac{p}{1-p}\cdot \sigma\sqrt{2\pi}\cdot e^{\frac{I\Delta^2}{2\sigma^2}}\\ 
       &= \frac{(1-w_{mn})\cdot e^{\epsilon-\frac{I\Delta^2}{2\sigma^2}}}{\sigma\sqrt{2\pi}+w_{mn}\cdot e^{\epsilon-\frac{I\Delta^2}{2\sigma^2}}}\cdot \sigma\sqrt{2\pi}\cdot  e^{\frac{I\Delta^2}{2\sigma^2}}\\
       &\leq \frac{(1-w_{mn})\cdot e^{\epsilon-\frac{I\Delta^2}{2\sigma^2}}}{\sigma\sqrt{2\pi}}\cdot \sigma\sqrt{2\pi}\cdot  e^{\frac{I\Delta^2}{2\sigma^2}}\\
       &=(1-w_{mn})\cdot e^{\epsilon}<e^{\epsilon}.
    \end{aligned}
    \end{equation}
    \end{enumerate}
\end{proof}
From equation (\ref{10}) and (\ref{12}), although the noise reduction achieved by Algorithm \ref{alg:alg2} is difficult to quantify, the fact that it offers stronger privacy protection than Algorithm \ref{alg:alg1} is evident as the existence of the term $(1-w_{mn})$.

\section{Evaluations} \label{sec:Evaluations}
To provide a fair and comprehensive assessment of the proposed mechanisms, we conduct extensive experiments across various models and datasets. The results are compared with several existing mechanisms, offering detailed insights into their relative performance.

\subsection{Setup}
\subsubsection{Datasets and Tasks}
In this study, we carefully select representative learning tasks from diverse domains, including computer vision, text mining, and data mining, where the datasets are likely to be sensitive. Given the confidentiality and proprietary nature of the data, it is more secure for data owners if this information is perturbed before being transmitted to the untrusted central server. For example, raw data or features extracted from raw data could contain sensitive information (e.g., personal images or house numbers), and this information should be carefully preserved before being transmitted to third parties. Another example is that the model parameters that are trained on sensitive data can also be sensitive, and the deep learning models computed on such data should not reflect it.

For the computer vision task, we select MNIST~\cite{lecun1998gradient}, CIFAR-10~\cite{krizhevsky2009learning}, and SVHN~\cite{netzer2011reading} as the datasets. MobileNet (pre-trained on ImageNet) \footnote{\url{https://pytorch.org/vision/main/models/generated/torchvision.models.mobilenet_v2.html#torchvision.models.MobileNet_V2_Weights}} is used for the computer vision task. The model is trained with the Adam optimizer with a learning rate of 0.001. It uses the cross-entropy loss and has a batch size of 64. We train it for 10 epochs. The MNIST dataset consists of 60,000 training samples and 10,000 test samples of 28$\times$28 pixel grayscale images of handwritten digits. CIFAR-10 contains 50,000 training samples and 10,000 test samples and has 10 classes of common objects, e.g., birds, airplanes, and cats. SVHN includes door number images collected from Google Street View, and it has 73,257 training images and 26,032 testing images.

In the context of data mining, we employ the DarkNet dataset~\cite{habibi2020didarknet}, which is designed to detect and characterize VPN and Tor applications. It combines two public datasets from the Canadian Institute for Cybersecurity (CIC), namely ISCXTor2016 and ISCXVPN2016. The Tor data is divided into 200 segments, which are then interspersed with non-Tor data. The final 20\% of the combined dataset is used for testing, with the remaining 80\% used for training, i.e., 113,184 training samples and 28,297 test samples. For this task, a Fully Connected Network (FCN) is employed, which consists of four hidden layers, with the first three having a hidden size of 256 and the last layer having a hidden size of 128. A dropout rate of 0.3 is applied. The model is trained using the Adam optimizer with a learning rate of 0.001. It also uses the cross-entropy loss with a batch size of 64. We train it for 25 epochs.

For the text mining task, we work with the IMDB dataset~\cite{maas2011learning}, which consists of 50,000 movie reviews used for binary classification. Each review is tokenized using the spaCy tokenizer, and each token is converted to its corresponding vocabulary index using pre-trained GloVe embeddings, with each review truncated or padded to a fixed length of 500 tokens. Both the training and test sets consist of 25,000 samples each, and the data is directly imported from Torchtext. A BiLSTM model with two hidden layers, each with a hidden size of 256, is employed with a dropout rate of 0.3. The model is trained using the Adam optimizer with a learning rate of 0.001. The cross-entropy loss is utilized for model training with a batch size of 64, and we train the model for 50 epochs.

\subsubsection{Baselines and Metrics} 
We compare TLDP with other differential privacy mechanisms designed for high-dimensional data. The baselines contain: the Laplace mechanism (Laplace), Gaussian mechanism (Gaussian), Matrix-Variate Gaussian (MVG)~\cite{chanyaswad2018mvg}, Independent Directional Noise (IDN)~\cite{direction}, Layer‐Level Adaptive Gradient Perturbation (LLAGP)~\cite{LLAGP}, and Differential Privacy Adaptive Gradient Clipping  Method (DPHSGD)~\cite{DPHSGD}. 
We implement TLDP-L, TLDP-G, TLDP-L-w, and TLDP-G-w as solutions for applying local differential privacy on private data (see Section~\ref{sec:methodologies}), where TLDP-L and TLDP-G denote TLDP that uses the Laplace and Gaussian noise, while TLDP-L-w and TLDP-G-w represent TLDP when using Laplace or Gaussian noise while applying a weight matrix for improving utility.
Unlike previous studies where noise designed for scalar-valued queries is directly applied to tensors, we mathematically prove that the Laplace and Gaussian noise we add satisfies LDP for tensors (Definition~\ref{def2}), as detailed in Theorem~\ref{thm:important}. IDN, on the other hand, is a $(\epsilon, \delta)$-differential privacy mechanism for tensor-valued queries. In all experiments, the utility subspace $W$ is set to the identity matrix $E$. {LLAGP is designed for adjusting the privacy budget at the layer‐level during training progress, while DPHSGD is for adjusting the clipping threshold of training gradients. Tailored to large-scale data, both of the last two methods work perfectly for conducting comparative tests.}

In our experiments, we assess the performance of classification models using a weighted F1-score (i.e., F1-score in the following text) for fair comparison over various tasks, including tasks with imbalanced distributions.

The F1-score is calculated as follows:
\begin{equation}
        P_i = \frac{TP_i}{TP_i+FP_i}, \\  
\end{equation}
\begin{equation}
    R_i = \frac{TP_i}{TP_i+FN_i}, \\  
\end{equation}
where $TP$ (True Positives) means the number of instances where the model correctly predicts the positive class. $FP$ (False Positives) is the number of instances where the model incorrectly predicts the positive class (i.e., it predicts positive when the actual class is negative). $FN$ (False Negatives) is the number of instances where the model incorrectly predicts the negative class (i.e., it predicts negative when the actual class is positive).
\begin{equation}
    F1_i = 2 \times \frac{P_i \times R_i}{P_i+R_i}, \\  
\end{equation}
\begin{equation}
    \textit{F1-score} = \frac{\sum_{i=1}^n(w_i \times F1_i)}{\sum_{i=1}^n w_i}, \\  
\end{equation}
where $w_i = \frac{N_c}{N}$. Here, $N_c$ and $N$ denote the number of samples of class $c$ and the total number of samples from all the classes, respectively.

\subsubsection{Privacy-Preserving Targets} 
In our experiments, we simulate the distributed computation system using three users and one central server. The users compute their local models using their private dataset and transmit the trained model parameters to the central server. The central server aggregates the received model parameters, computes an average of them, and sends it back to all the users. We denote the models trained with each user's dataset as local models and the aggregated model on the central server as the global model. Here, we consider three typical tasks in the distributed computation system, i.e., original data, training features, and model parameters. For the first type, users simply transmit their private data to the server, leveraging the enormous computing resources of the server to process and analyze the data. The second type focuses on features extracted from private datasets. This scenario frequently arises in distributed model training across multiple parties, where users process raw data to extract features and transmit the extracted feature embeddings instead of raw data to the central server for model training to avoid privacy leakage. The query mechanism outputs intermediate data/features with differential privacy guarantees, which are then utilized for subsequent training processes on the server side. The third type emphasizes the protection of model parameters. Sharing model parameters is very common in distributed systems, especially when federated learning \cite{mcmahan2017communication} has evolved for model computation without the need for direct data or feature sharing between data owners and a parameter server. However, research~\cite{geiping2020inverting} has found that attackers could easily infer the original data used to train the model from intercepted model parameters. Given that the trained model parameters are sensitive, the query function returns a differentially private version of them before transmitting.
We present the implementation details categorized by type as follows.

\subsection{Implementation Details} \label{subsec:ImplementationDetails}
\subsubsection{Type I: Private Testing Data}
The datasets used in Type I include MNIST, CIFAR-10, and SVHN. The setup is illustrated in Table~\ref{tab:table2}. In this type, the training data is likely to be private and should be preserved. Users manipulate their respective data and send it directly to a central server, which uses this information to train the model parameters before transmitting the updated parameters back. No preprocessing or size-trimming is applied to the raw data, ensuring that the accuracy of the experimental results is maintained at the highest possible level.

Additionally, for the MNIST dataset, we include comparisons with the TLDP-L-w and TLDP-G-w methods. In these approaches, the weights are determined based on the grayscale values of the images: higher grayscale values correspond to higher weights. This setting is reasonable since, as the dataset has been preprocessed, apart from the digit regions with values greater than 0, all other parts are set to 0.
        
\textit{Query function:} The query function we use is the identity function $f(\mathcal{X})=\mathcal{X}$, where $\mathcal{X}$ represents the original tensors from users. 
The $l_2$ sensitivity for neighboring datasets $\{\mathcal{X},\mathcal{X}'\} $ is defined as:
\begin{equation}
    s_2(f)=\underset{\mathcal{X},\mathcal{X}'}{sup}||\mathcal{X}-\mathcal{X}'||_F=2\sqrt{I_1I_2\dots I_N},
\end{equation}
since the range of data is set to $(-1,1)$.

\begin{table}[h]
\caption{Datasets and parameter settings for Type I}\label{tab:table2}
\centering
\resizebox{0.8\columnwidth}{!}{

\begin{tabular}{ c c m{1cm} m{1cm} m{1cm} c }
\hline
Dataset &Model&Training Data&Testing Data&Batch size&$\Delta$\\
\hline
MNIST&MobileNet&60000&10000&64&256\\
CIFAR-10&MobileNet&50000&10000&64&256\\
SVHN&MobileNet&73257&26032&64&256\\
\hline
\end{tabular}

}
\end{table}

\subsubsection{Type II: Private Training Features}
For Type II, we select MNIST, CIFAR-10, and SVHN as experimental datasets. The setup is shown in Table~\ref{tab:table3}. DarkNet and IMDB are not used in this task due to the absence of pre-trained models, and without them, the extracted feature embeddings are prone to being less meaningful. We consider the features extracted from the original training data of these datasets to be confidential and are committed to protecting privacy. To enhance model performance, we substitute the activation function in each layer of MobileNet from ReLU6 to tanh$(\cdot)$, ensuring that the output features are normalized within the range of $(-1, 1)$. MobileNet is divided into two parts: the feature extractor (i.e., all layers except the final fully connected layer) and the classifier (i.e., the final fully connected layer). After extracting features from the feature extractor, noise is added to the features, and only the classifier is trained and tested.  

\textit{Query function:} We use the identity function $f(\mathcal{X})=\mathcal{X}$ as the query function, and $\mathcal{X}$ represent features extracted from the original data. In TVG, for neighboring datasets $\{\mathcal{X},\mathcal{X}'\} $, the $l_2$ sensitivity is:
\begin{equation}
    s_2(f)=\underset{\mathcal{X},\mathcal{X}'}{sup}||\mathcal{X}-\mathcal{X}'||_F=2\sqrt{I_1I_2\dots I_N},
\end{equation}
since the range of data is set to $(-1,1)$.
\begin{table}[h]
\caption{Datasets and parameter settings for Type II \label{tab:table3}}
\centering
\resizebox{0.8\columnwidth}{!}{
\begin{tabular}{ c c m{1cm} m{1cm} m{1cm} c }
\hline
Dataset &Model&Training Data&Testing Data&Batch size&$\Delta$\\
\hline
MNIST&MobileNet&60000&10000&64&256\\
CIFAR-10&MobileNet&50000&10000&64&256\\
SVHN&MobileNet&73257&26032&64&256\\
\hline
\end{tabular}

}
\end{table}

\subsubsection{Type III: Private SGD}
For this task, we utilized datasets including MNIST, CIFAR-10, SVHN, IMDB, and DarkNet, as shown in Table~\ref{tab:table4}. The training data is split into three parts, and each local model is trained on one subset. After completing a single epoch, each local model uploads its weights and biases using the global model. The global model then updates its parameters by averaging the weights and biases from the three local models and subsequently sends the updated parameters back to the local models. As a result, after each iteration, the weights and biases of all four models are synchronized. The entire process can be restated as follows:
\begin{itemize}
    \item Train each local model on its subset. 
    \item Clip the model parameters by their $L_\infty$ norm, with a clipping threshold of $C$.
    \item Calculate the mean of the model parameters across a batch of subsets and introduce noise to these averaged model parameters to ensure privacy. 
    \item Update the model parameters using these noise-adjusted model parameters and repeat the process starting from step 1.
\end{itemize}

\textit{Query function:} 
The query function we use is the identity function $f(\mathcal{X})=\mathcal{X}$, 
and $\mathcal{X}$ represents the model parameters trained by users. 
Similarly, we define $l_2$ sensitivity as, for neighboring datasets 
$\{\mathcal{X},\mathcal{X}'\}:$
\begin{equation}
    s_2(f)=\underset{\mathcal{X},\mathcal{X}'}{sup}||\mathcal{X}-\mathcal{X}'||_F=2C\sqrt{I_1I_2\dots I_N}
\end{equation}
where C is the clip value mentioned above.

\begin{table}[h]
\caption{Datasets and parameter settings for Type III}
\label{tab:table4}
\centering
\resizebox{\linewidth}{!}{
\begin{tabular}{ c c m{1cm} m{1cm} m{1cm} m{1cm} c}
\hline
Dataset &Model&Training Data&Testing Data&Batch size& Clip value&$\Delta$\\
\hline
MNIST&MobileNet&60000&10000&64&1&256\\
CIFAR-10&MobileNet&50000&10000&64&1&256\\
SVHN&MobileNet&73257&26032&64&1&256\\
IMDB&BiLSTM&25000&25000&64&1&400001\\
DarkNet&MLP&113184&28297&64&1&256\\
\hline
\end{tabular}
}
\end{table}

\subsection{Experimental Results}
Before comparing the experimental outcomes, we provide a theoretical analysis of the expected error associated with each (local) differential privacy approach. The details of these theoretical findings are summarized in Table~\ref{tab:table5}. 
\begin{table*}[t]
\caption{COMPARISON OF EXPECTED ERROR} \label{tab:table5}
\centering
\begin{tabular}{c c c c c c c c}
\hline
 Method & Laplace & Gaussian & MVG & IDN &DPHSGD & TLDP-L & TLDP-G\\
\hline
$\mathbb{E}||\mathcal{Z}||_2$ & 
$\frac{\Delta}{\epsilon}I\sqrt{2I}$ & $\frac{\Delta}{2\epsilon}I\sqrt{I}$ & 
$\frac{\Delta I_1 I \sqrt{I}}{\sqrt{2}\epsilon}\cdot \ln^2(I_1+1)$&
$\frac{\Delta I\sqrt{I}}{\sqrt{2}\epsilon}$&
$\frac{C\Delta}{\epsilon}I\sqrt{2I}$&
$\frac{\Delta}{\epsilon}\sqrt{2(1-p)I}$ & 
$\frac{\Delta}{2\epsilon}\sqrt{(1-p)I}$\\
\hline
\end{tabular}
\end{table*}

In Table~\ref{tab:table5}, $I=I_1I_2\dots I_N$ is a product of each dimension of the tensor. Since the query functions for the three tasks outlined in Section~\ref{subsec:ImplementationDetails} are identical, the expected errors are consistent, which is defined as 
\begin{equation}
    \underset{\mathcal{Z}\in\mathbb{R}^{I_1I_2\dots I_N}}{sup}\mathbb{E}||\mathcal{Z}||_2=\underset{\mathcal{Z}\in\mathbb{R}^{I_1I_2\dots I_N}}{sup}\mathbb{E}(\sum_{i=1}^I(\mathcal{Z}_i)^2)^\frac{1}{2},
\end{equation}
where $\mathcal{Z}_i (i=1,2,\dots,N)$ denote the components of $\mathcal{Z}$. Using the properties of variance, we can derive that for TLDP-L noise, we have
\begin{equation}
    \begin{aligned}
        [\mathbb{E}(\sum_{i=1}^I&(\mathcal{Z}_i)^2)^\frac{1}{2}]^2\leq\mathbb{E}[(\sum_{i=1}^I(\mathcal{Z}_i)^2)^\frac{1}{2}]^2=\mathbb{E}(\sum_{i=1}^n(\mathcal{Z}_i)^2)\\
        &=\sum_{i=1}^I\mathbb{E}(\mathcal{Z}_i)^2=\sum_{i=1}^I(1-p)2b^2=I(1-p)2b^2\\
        &=\frac{2\Delta^2}{\epsilon^2}\cdot (1-p)I,
    \end{aligned}
\end{equation}
since for noise following $Lap(0,b)$, its expectation is 0 and variance is $2b^2$. So
\begin{equation}
    \underset{\mathcal{Z}\in\mathbb{R}^{I_1I_2\dots I_N}}{sup}\mathbb{E}||\mathcal{Z}||_2=\frac{\Delta}{\epsilon}\sqrt{2(1-p)I}.
\end{equation}
For TLDP-G noise,
\begin{equation*}
    \begin{aligned}
        [\mathbb{E}(\sum_{i=1}^I&(\mathcal{Z}_i)^2)^\frac{1}{2}]^2\leq\mathbb{E}[(\sum_{i=1}^I(\mathcal{Z}_i)^2)^\frac{1}{2}]^2=\mathbb{E}(\sum_{i=1}^n(\mathcal{Z}_i)^2)\\
        &=\sum_{i=1}^I\mathbb{E}(\mathcal{Z}_i)^2=\sum_{i=1}^I(1-p)\sigma^2=I(1-p)\sigma^2\\
        &=\frac{\Delta}{4\epsilon^2}(1-p)I,
    \end{aligned}
\end{equation*}
since for noise following $\mathcal{N}(0,\sigma^2)$, its expectation is 0 and variance is $\sigma^2$. So
\begin{equation}
    \underset{\mathcal{Z}\in\mathbb{R}^{I_1I_2\dots I_N}}{sup}\mathbb{E}||\mathcal{Z}||_2=\frac{\Delta}{2\epsilon}\sqrt{(1-p)I}.
\end{equation}

If the randomized response mechanism is removed from our method, i.e., no component will be retained with the probability $p$, and the term $1-p$ in the outcomes will become $1$. Moreover, as discussed in the proof of Theorem~\ref{thm:important}, in order to ensure the same privacy budget, the values of the two noise parameters are increased by a factor of $I$ compared to the improved methods. As a result, the expected errors of normal Laplace and Gaussian mechanisms are $\frac{\Delta}{\epsilon}I\sqrt{2I}$ and $\frac{\Delta}{2\epsilon}I\sqrt{I}$, respectively.

For MVG and IDN, we adopt the same parameters as Task II in~\cite{chanyaswad2018mvg} and~\cite{direction}, which provides the expected error $$\mathbb{E}_{MVG}||\mathcal{Z}||_2=\sqrt{\frac{I_1 I}{(-\beta_0 + \sqrt{\beta_0^2+8\alpha_0\epsilon})^2}\cdot4 \alpha_0^2}$$ and $$\mathbb{E}_{IDN}||\mathcal{Z}||_2=\sqrt{\frac{Is_2^2(f)}{(-\zeta(\delta)+\sqrt{\zeta^2(\delta)+2\epsilon})^2}},$$ where $\zeta^2(\delta)=-2\ln{\delta}+2\sqrt{-I\ln{\delta}}+I$. Removing the relaxation term $\delta$, we obtain its estimated value $\frac{\Delta I_1 I \sqrt{I}}{\sqrt{2}\epsilon}\cdot \ln^2(I_1+1)$ and $\frac{\Delta I\sqrt{I}}{\sqrt{2}\epsilon}$.  
Compared with the Gaussian mechanism, DPHSGD just adjusts the clipping threshold of parameters essentially, thus the difference between their expected errors is only a finite constant $C$. However, the expected error of LLAGP cannot be compared with other methods since the noise is injected into hidden layers instead of the tensors directly.

A comparison of the results in Table \ref{tab:table5} shows that whether Laplacian noise or Gaussian noise is injected, the expected error of TLDP is $O(1/I)$, which is consistent with the conclusions of Theorem \ref{thm:important}.

Next, we present the experimental results for privacy algorithms and assess their consistency with the theoretical error. 

\begin{figure*}[t]
\centering
\includegraphics[width=\textwidth]{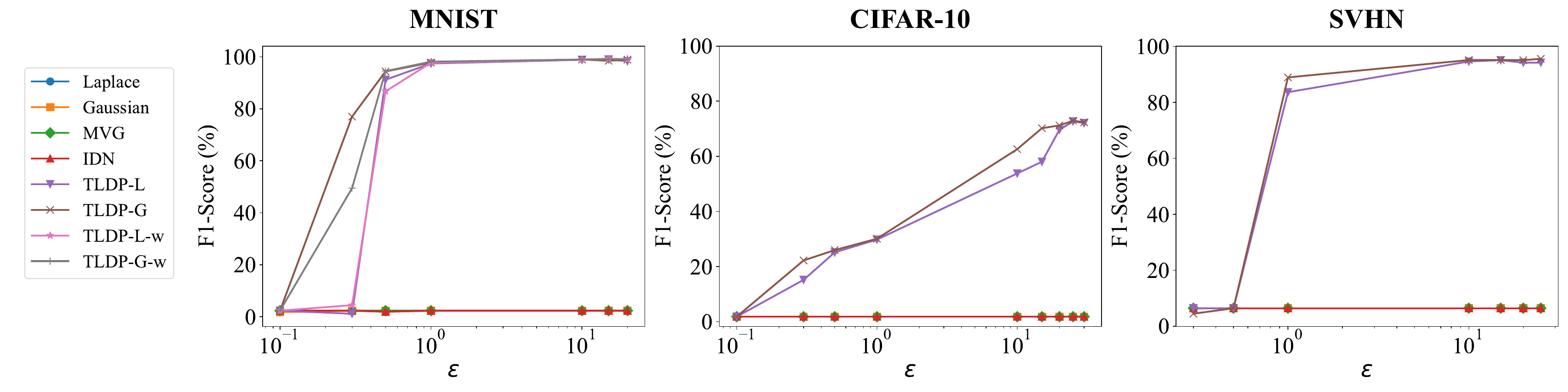}
\caption{Type I: Private Testing Data. The TLDP-G mechanism demonstrates a slight advantage over the TLDP-L mechanism, with both surpassing the other baselines by a significant margin. This is especially evident in the MNIST and SVHN datasets, where the F1-score of TLDP rapidly increases to a high level.} \label{fig:type1}
\end{figure*} 

\begin{figure*}[t]
\centering
\includegraphics[width=\textwidth]{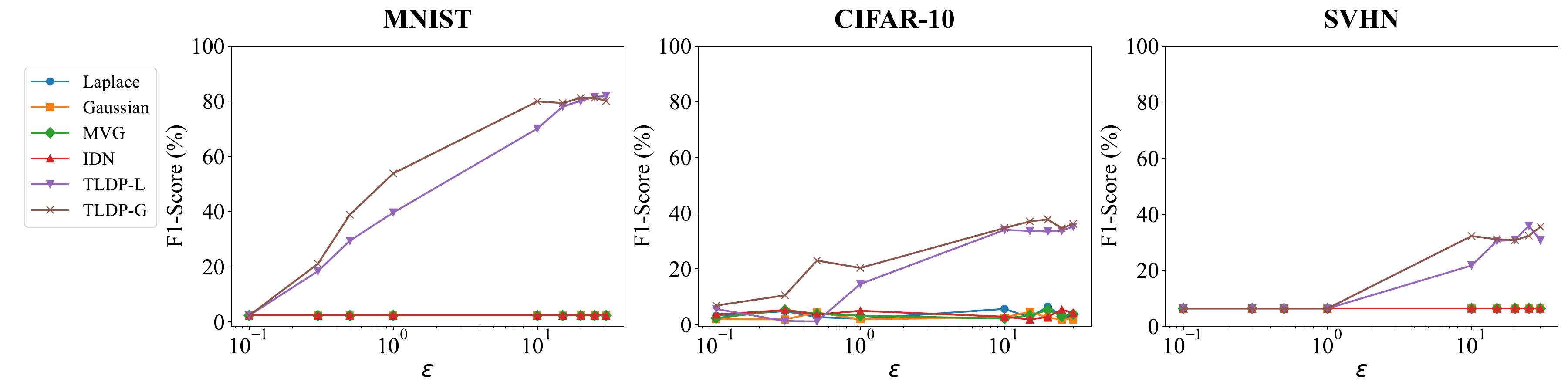}
\caption{Type II: Private Training Features. Both TLDP methods consistently outperform the other approaches, which are largely ineffective at completing any meaningful recognition tasks.} \label{fig:type2}
\end{figure*}

\begin{figure*}[t]
\centering
\includegraphics[width=\textwidth]{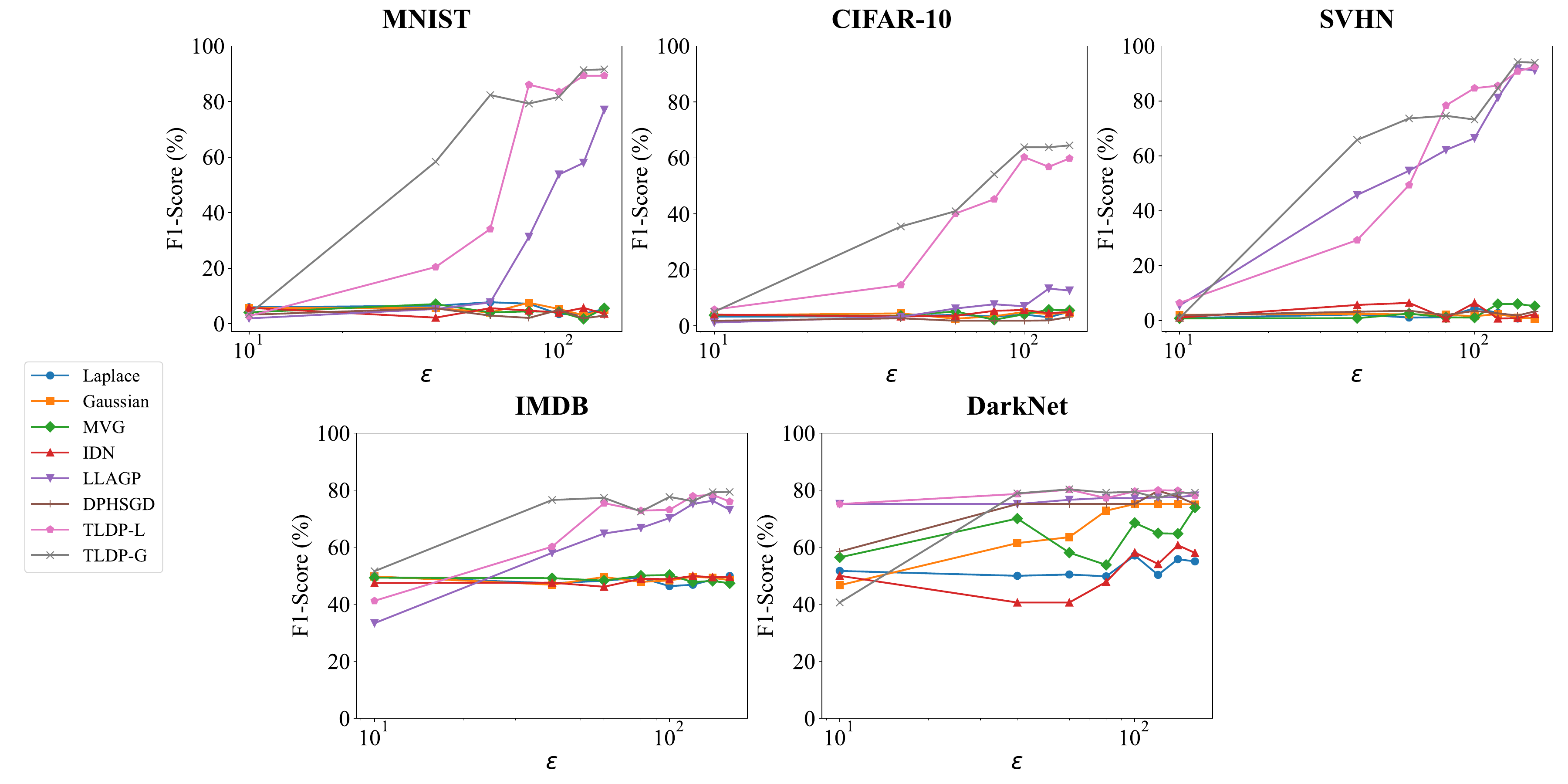}
\caption{Type III: Private Stochastic Gradient Descent. The TLDP methods continue to demonstrate unparalleled advantages. However, in certain datasets, as $\epsilon$ increases, the performance of TLDP-L temporarily surpasses that of TLDP-G. Ultimately, the performance of both methods tends to converge.} \label{fig:type3}
\end{figure*}

\subsubsection{Type I Results} 
Fig.~\ref{fig:type1} summarizes the results of Type I across various computer vision datasets, i.e., MNIST, CIFAR-10, and SVHN. It demonstrates that in computer vision tasks, our TLDP algorithm can achieve much higher classification accuracy than the comparative methods because the latter introduces excessive noise to the tensor data, making accurate identification nearly impossible.

Fig.~\ref{fig:type1} reports F1-score under different $\epsilon$, and the relaxation term $\delta$ is fixed as $10^{-5}$, similar to~\cite{direction}. The experimental results across three datasets clearly demonstrate that as $\epsilon$ increases, the F1-scores of the Laplace, Gaussian, MVG, and IDN mechanisms remain consistently low, indicating their inability to effectively complete the task. In contrast, the F1-score of the TLDP method rises sharply. On the MNIST and SVHN datasets, the score reaches approximately 90\% when $\epsilon$ is below 1. Although the growth is slower on the CIFAR-10 dataset, it eventually stabilizes around 70\%. For the MNIST dataset, we also included the TLDP-L-w and TLDP-G-w methods. While the performance of the former is nearly identical to that of TLDP-L, the latter slightly underperforms TLDP-G, which is consistent with our earlier theoretical analysis. The addition of a weight matrix ensures better privacy protection under the same differential privacy budget at the cost of slightly reduced accuracy. Furthermore, the TLDP-G mechanism outperforms the TLDP-L mechanism, which is in line with our previous theoretical error analysis.

\subsubsection{Type II Results} 
The results of Type II for each dataset are presented in Fig.~\ref{fig:type2}. Due to the fact that the central server now operates on extracted feature tensors instead of raw data in the training set, the performance metrics of the TLDP method experience a decline. 
However, despite this reduction, the TLDP method still significantly outperforms the Laplace, Gaussian, MVG, and IDN mechanisms. 

Fig.~\ref{fig:type2} presents the F1-score across different values of $\epsilon$, where the relaxation term $\delta$ is set to $10^{-5}$. Compared to Type I, the rate of increase in F1-score across the three experimental results slows down as $\epsilon$ increases in Type II. For the MNIST dataset, the F1-score reaches a maximum of about 80\%, while in CIFAR-10 and SVHN, the increase is much more limited, with F1-scores hovering around 40\% finally. Unlike in Type I, the performance of the Laplace, Gaussian, MVG, and IDN mechanisms in Type II shows occasional fluctuations in F1-scores, although they still remain at extremely low levels, indicating their ineffectiveness for classification tasks. Additionally, consistent with the theoretical error analysis presented earlier, the TLDP-G mechanism generally outperforms the TLDP-L mechanism in these experiments.

\subsubsection{Type III Results} 
Fig.~\ref{fig:type3} shows the results of Type III on the different datasets for comparison. Two new datasets, IMDB for text mining and DarkNet for data mining, and two new comparative methods, LLAGP and DPHSGD, are introduced in the experiment. In Type III, due to the transmission of more abstract model parameters to the central server, the results from the various methods only become comparable when $\epsilon$ exceeds 100. As anticipated, the performance of the Laplace, Gaussian, MVG, and IDN mechanisms remains significantly low. {LLAGP outperforms all other comparative methods across most scenarios except on the CIFAR-10 dataset, but still fails to surpass the method proposed in this paper, while DPHSGD only delivers good performance on the Darknet dataset and performs poorly in all other cases. Overall, the TLDP method proposed in this paper achieves the best comprehensive performance.} 

Fig.~\ref{fig:type3} displays the F-1 score for varying $\epsilon$ values, with the relaxation term $\delta$ fixed as $10^{-5}$. In Type III, the results for the MNIST, CIFAR-10, and SVHN datasets are consistent with those in Type I and Type II, with TLDP significantly outperforming the Laplace, Gaussian, MVG, IDN, and DPHSGD mechanisms. LLAGP performs best on the SVHN dataset, and can achieve accuracy close to that of TLDP when $\epsilon$ is sufficiently large; however, it performs worst on the CIFAR-10 dataset, with results that are not significantly different from those of other comparative experiments. On the IMDB dataset, the accuracy of the Laplace, Gaussian, MVG, and IDN mechanisms improves to around 40\%, but no further improvement is observed as $\epsilon$ increases. In contrast, both TLDP methods consistently achieve accuracy levels of around 80\%. LLAGP also performs well on this dataset, keeping pace with TLDP methods. On the DarkNet dataset, the performance of the different methods is more distinctly differentiated. The IDN and Laplace mechanisms show the poorest performance, followed by the Gaussian and MVG mechanisms. Although the TLDP method continues to demonstrate an advantage, the gap with other methods is less pronounced compared to the former experiments—especially since the accuracy of LLAGP and DPHSGD has ultimately become sufficiently close to that of TLDP. This is especially true for the TLDP-G mechanism, which initially exhibits the lowest accuracy when $\epsilon$ is set to 10 but eventually surpasses all other methods, aligning with the TLDP-L curve.

We find that the performance of IDN is not as good as claimed in~\cite{direction}, primarily due to its high reliance on its pre-defined utility subspace. When there is a lack of known conditions, its performance is significantly compromised. Most comparative methods (i.e., Gaplace, Gaussian, MVG, IDN) perform poorly in experiments, mostly because the amount of noise they inject is $O(I)$ times that of the TLDP method (where $I=784$ for the MNIST dataset and $I=1024$ for CIFAR-10 and SVHN datasets), and the larger value of $I$ results in the poorer performance of the comparative methods.
It is worth noting that LLAGP performs better than Laplace, Gaussian, MVG, IDN, and DPHSGD; the underlying reason behind this could be that it selectively allocates less privacy budget to the hidden layers close to the output layer and more privacy budget to those that are close to the input layer to avoid privacy leakage (under its assumption that these hidden layers face a higher privacy leakage risk compared with the rest of the layers) instead of treat all hidden layers equally.  

The results from all the aforementioned experiments provide strong evidence of the robustness of the TLDP mechanism across diverse experimental settings. Regardless of the dataset or the tasks, TLDP consistently outperforms other baseline methods, showcasing its ability to adapt effectively to different environments. This robustness highlights TLDP's potential as a reliable solution for preserving privacy in distributed computing systems while maintaining high model performance, making it a promising candidate for a wide range of distributed applications in real-world scenarios.

\section{Conclusion}
In this paper, we introduce TLDP, a novel algorithm for implementing LDP on tensor data, which enhances traditional Laplace and Gaussian differential privacy mechanisms. TLDP is built upon the randomized response technique, where the original data is retained with a certain probability, and noise is added otherwise. This approach replaces the direct interference of noise with randomness, effectively addressing the issue of excessive noise when applying LDP to high-dimensional data. We provide a unified definition of localized differential privacy tailored to tensor data and rigorously prove that the TLDP mechanism satisfies this definition. To further improve the utility, we incorporate a weight matrix into the algorithm, reducing the amount of noise and thereby achieving higher accuracy under the same privacy budget. Finally, we compare our method with several existing privacy-preserving algorithms, both in terms of theoretical error analysis and empirical evaluations conducted on privacy-sensitive datasets with identity query functions. The results from both the theoretical and experimental evaluations demonstrate that TLDP consistently outperforms existing approaches in terms of accuracy and overall performance.

\section*{Acknowledgments}
This work was supported in part by the National Natural Science Foundation of China (62406215) and the Science and Technology Program of Jiangsu Province (BZ2024062).

\bibliographystyle{elsarticle-num}
\bibliography{main1}

\end{document}